\documentclass{amsart}%
\usepackage{amssymb,amsmath,amsfonts}
\usepackage{verbatim}
\usepackage{enumerate}
\usepackage{ifpdf}
\usepackage{graphicx}
\graphicspath{{pics/}}
\newcommand{\x}{\mathbf{x}}
\newcommand{\y}{\mathbf{y}}

\usepackage{amssymb,amsmath,amsfonts,amsthm}

 \usepackage{cite}
\usepackage{verbatim}
\usepackage{array}
\usepackage{tensor}
\usepackage[pdftex]{hyperref}
\usepackage{epsfig}
\usepackage{tikz}
\usepackage{lipsum}
\usetikzlibrary{shapes.geometric}

\usepackage{tweaklist}
\renewcommand{\enumhook}{\setlength{\topsep}{-0.5pt}%
  \setlength{\itemsep}{-0.5pt}}
\newcommand{\tn}{\textnormal}

\newcommand{\bra}[1]{\mbox{$\langle #1|$}}
\newcommand{\ket}[1]{\mbox{$|#1\rangle$}}
\newcommand{\braket}[2]{\mbox{$\langle #1|#2\rangle$}}

\DeclareMathOperator{\Tr}{Tr}

\newcommand{\NP}{{\sf{NP}}}
\renewcommand{\P}{{\sf{P}}}

\newcommand{\be}{\begin{equation}}
\newcommand{\ee}{\end{equation}}

\newtheorem{theorem}{Theorem}
\newtheorem{lemma}[theorem]{Lemma}
\newtheorem{corollary}[theorem]{Corollary}
\newtheorem{remark}[theorem]{Remark}

\newtheorem{definition}[theorem]{Definition}
\newtheorem{example}[theorem]{Example}

\newtheorem{proposition}[theorem]{Proposition}

\definecolor{dred}{rgb}{.8,0.2,.2}
\definecolor{ddred}{rgb}{.8,0.5,.5}
\definecolor{dblue}{rgb}{.2,0.2,.8}
 
\def\1#1{{\bf #1}} 
\def\2#1{{\cal #1}}
\def\3#1{{\sl #1}} 
\def\4#1{{\tt #1}}
\def\5#1{{\sf #1}}
\def\6#1{{\mathfrak #1}}
\def\7#1{{\mathbb #1}}

\newcommand{\CC}{\mathcal{C}}
\makeatletter
\g@addto@macro{\endabstract}{\@setabstract}
\newcommand{\authorfootnotes}{\renewcommand\thefootnote{\@fnsymbol\c@footnote}}%
\makeatother

\graphicspath{{pics/}}
\begin{document}
\title{Tensor Network Contractions for \#SAT}
\author[Jacob D.~Biamonte]{Jacob D.~Biamonte\textsuperscript{1}
}

\author[Jason Morton]{Jason Morton\textsuperscript{2}
}
\author[Jacob Turner]{Jacob Turner\textsuperscript{1,2}
}
 \maketitle
 \begin{center}
 \textsuperscript{1}\small{ISI Foundation, Via Alassio 11/c 10126 Torino Italy}
\par
   \textsuperscript{2}\small{Department of Mathematics, The Pennsylvania State
University, University Park, PA 16801 USA}\par \bigskip
\end{center}

\begin{abstract}
The computational cost of counting the number of solutions satisfying a
 Boolean formula, which is a problem instance of \#SAT, has proven subtle to quantify. Even when finding individual satisfying solutions is computationally easy (e.g. 2-SAT, which is in $\P$), determining the number of solutions is \#$\P$-hard. Recently, computational methods simulating quantum systems experienced advancements due to the development of tensor network algorithms and associated quantum physics-inspired techniques.  By these methods, we give an algorithm using an axiomatic tensor contraction language for $n$-variable \#SAT instances with complexity  $O((g+cd)^{O(1)} 2^c)$ where $c$ is the number of COPY-tensors, $g$ is the number of gates, and $d$ is the maximal degree of any COPY-tensor.  Thus, counting problems can be solved efficiently when their tensor network expression has at most $O(\log c)$ COPY-tensors and polynomial fan-out. This framework also admits an intuitive proof of a variant of the Tovey conjecture (the r,1-SAT instance of the Dubois-Tovey theorem). This study increases the theory, expressiveness and application of tensor based algorithmic tools and provides an alternative
insight on these problems which have a long history in statistical physics and
computer science. \end{abstract}

\section{Introduction}
Given any property used as a measure of complexity, there are two natural questions: How does this property affect the computation of the problem and how easy is it to determine the value of said property. In this paper, we discuss several such properties and address these two questions with regard to counting K-SAT problems. These properties include 
\begin{enumerate}
\item Bipartitions of tensor networks,
 \item Tensor networks with the structure of a tree,
 \item The fan-out of the variables,
 \item Treewidth, branchwidth, and clique-width, and
 \item  R\'enyi entropy of tensor networks.
\end{enumerate}
 We measure the fan-out of the variables by the number of
COPY-tensors in the tensor network. The resulting algorithm we describe gives an upper bound dependent on
these tensors.
 We also discuss the R\'enyi entropy as a measure of complexity, especially with
regard to the fact that it may not be defined or easy to find.\\

 \noindent\textbf{Background.}
Satisfiability problems and
counting problems have several physical analogues and have had a long and
fruitful history intersecting physics \cite{Barahona,Kirkpatrick}. 
In 1982 Barahona proved that finding the lowest energy state of Ising
spins acting on a 2D lattice is $\NP$-hard \cite{Barahona}. Other seminal
results include the
concepts of algorithmic simulated annealing \cite{Kirkpatrick} which has long been applied to K-SAT and related
problems. 

The K-SAT problem on $n$ variables is formed from a conjunction of $c$ clauses. 
Each clause is the disjunction of $K$ Boolean variables (appearing in
complimented or uncomplimented form). Several celebrated results have shown that the
ratio $c/n$ is an indicator of what appears to be an algorithmic phase
transition \cite{monasson1999determining}. 

Related studies have also unveiled indicators of such `critical
behavior' in other satisfiability problems \cite{goerdt1996threshold,kirkpatrick1994critical}. We discuss how the algorithm suggested in Theorem \ref{thm:upperbound} compares with treewidth based algorithms, especially in the critical case when the clause to variable ratio is close to one.

Further subtleties arise when we consider the counting versions of these
problems (\#2-SAT, \#3-SAT, or \#SAT in general). 
This version of the problem asks for the number of satisfying bitstrings of a given function $f$. 
For example, although 2-SAT is in $\P$, \#2-SAT
is $\#\P$-complete \cite{valiant1979complexity}. 

Tensor network methods are a collection of techniques to model and
reason about multilinear maps.  These methods form the backbone of tensor
network contraction algorithms to model quantum systems \cite{verstraete08,JCJ10, 2011JSP...145..891E, 2013arXiv1308.3318E} and
are used in the abstract tensor languages \cite{Penrose} to represent
channels, maps, states and processes appearing in quantum
theory~\cite{2005PhRvA..71d2337G, 2006PhRvA..73e2309G, 2007PhRvL..98v0503G,
CTNS}. 

There is body of recent work which has applied tensor
network methods appearing in physics to solve combinatorial problems appearing
in computer science \cite{JBCJ13, CTNS, CM12,
2013arXiv1302.1932M, KGE14, 2012PhRvA..86c0301M, wolf2011problems}. 
Contracting a network to determine the norm of a \emph{Boolean tensor network state} is equivalent to
counting the number of satisfying solutions of a Boolean formula. Thus in
general, tensor contraction is \#$\P$-hard \cite{Damm03}.

Boolean tensor network states have appeared in several
recent studies.  This includes the recent identification of several undecidable 
problems in tensor network states \cite{2012PhRvL.108z0501E, 
wolf2011problems, KGE14, 2012PhRvA..86c0301M}.  The use of Boolean states have
also suggested the use of
entropy as a measure of computational complexity \cite{CM13} as well as several
other interesting applications \cite{2008PhRvA..77e2331B, 2012JPhA...45a5309D, 2012EL.....9957004W, AL13, 2014arXiv1406.5279G}.

Relationships between multilinear maps and counting problems have long been known in computer science and physics \cite{valiant1979complexity2,Penrose}.  Several alternative approaches to the counting problem based on tensor network
methods have recently appeared \cite{CM12, GL12, 2014arXiv1406.5279G} as well as the
work done by the present authors \cite{JBCJ13, 2013arXiv1302.1932M}. Any problem in $\#\mathbf{\mathsf{P}}$ can be expressed as a \#SAT problem, which we express as a tensor network in a standard way \cite{JBCJ13}. In this paper, we exploit the expression of these problems as multilinear maps to give an algorithm for solving \#SAT problems based on the fan-out of the variables. If the fan-out is small, the algorithm is polynomial. We argue that this is a different measure than tree-width, to which it appears similar.
\newline 

\noindent {\bf Results.} 
Our contribution includes the discussion of tensor network properties and their effect on the
computational complexity of  tensor contraction. We use the fan-out to classify the complexity of an algorithm we propose. We discuss the use of R{\'e}nyi entropy as a measure of complexity. We also
prove several other results related to counting by tensor contraction.  
Theorem \ref{thm:rof} gives a an alternative proof
for the $r,1$-SAT case of the Tovey-Dubois theorem, proving that all read-once
formulas are satisfiable. We propose a tensor contraction algorithm based on the
number of COPY-tensors in the network and in Theorem \ref{thm:upperbound} give
the running time. This leads to a class of efficiently solvable \#SAT instances,
outlined in Corollary \ref{cor:efficientSAT}. We compare this to other algorithms based on graph theoretic measures, predominantly those defined in terms of treewidth. We also discuss the difficulty of
determining if R{\'e}nyi entropy is defined in Theorems \ref{thm:renyi1} and
\ref{thm:renyi2}.
\newline

\noindent {\bf Structure.} 
In Section \ref{sec:graphical} we recall the basic ideas
connecting
counting, satisfiability, and related physical problems. We then cast these
problems into the language of tensor network contraction.
In Section \ref{sec:sat-of-certain}, we prove satisfiability of read-once formulas
and then in Section \ref{sec:algorithm}, we present our tensor contraction
algorithm and compare it to other known algorithms. Lastly, in Section \ref{sec:renyi}, we discuss how the computational complexity relates to the R{\'e}nyi entropies of a Boolean state concluding with a discussion in Section \ref{sec:discussion}.  In Appendix \ref{sec:tensor-defs} we have included
the tensor network definitions required for the proofs appearing in this paper.

\section{Counting by tensor contraction}\label{sec:graphical}

 We first recall how
satisfiability and counting are connected to familiar concepts in
physics. Further, determining properties of many
physical systems
can be rephrased in terms of satisfiability and counting problems.  
%

Each $n$-variable 3-SAT problem can be associated to a $2^n\times 2^n$ Hamiltonian $H$ indexed by variable assignments.
Let
\begin{equation*}
 \psi_f := \sum_\x f(\x) \ket{\x} 
\end{equation*}
and 
\begin{equation*}
 \psi^\perp_f := \sum_\x (1-f(\x)) \ket{\x} = \sum_\x \overline{f}(\x) \ket{\x} = 
\ket{+}^{\otimes n} - \psi_f,
\end{equation*}
where $\ket{+}:= |0\rangle+|1\rangle$ and $\braket{\psi_f}{\psi_f^\perp} = 0$.   Then for $H
:= \ket{\psi^\perp_f}\bra{\psi_f^\perp}$, (note that for physical
considerations, we implicitly assume that $f$ is satisfiable or otherwise $H$
will be the zero matrix) we have that the eigenvectors of $H$ correspond to non-satisfying bitstrings.

Such observations have
been central to understanding the computational complexity of determining
physical properties of many naturally occurring spin
lattice problems \cite{Barahona}.  A typical problem includes determining a
ground state wave function of an Ising spin system (such as Kirpatrick et al.'s
simulated annealing \cite{Kirkpatrick}). This approach also leads to the
definition of the QMA complexity class and the local Hamiltonian problem, which
is QMA-complete \cite{kempe2006complexity}.

Then we note  
\begin{equation*}
 \min_{\psi} \braket{\psi}{H(\psi)} \geq 0
\end{equation*}
where $\psi$ is a bit string, and equality implies that $\psi$ encodes
a satisfying assignment to $f$. The ground space of $H$ is spanned by all
satisfying bitstrings, which is degenerate when there are multiple satisfying assignments. The canonical partition function 
\begin{equation*} 
 Z(H, \beta) := e^{-\beta H}
\end{equation*}
at the $\beta \rightarrow \infty$ limit (corresponding to the temperature $T \rightarrow 0$) in general has a degenerate ground state.  Therefore the number of satisfying assignments of $f$
(which we denote as $\#f$) satisfies the following bound:
\begin{equation*}
 0\leq \#f \leq  \text{Tr}\; Z(H, \beta).
\end{equation*}
Under this embedding the upper-bound holds with equality in the low-temperature limit of
the partition function, i.e.

\begin{equation}\label{eqn:lowtemp-count}
 \text{Tr} \lim_{\beta \rightarrow \infty} Z(H, \beta) = \#f.
\end{equation}
This can then be translated into a tensor contraction problem. Let 
\begin{equation*}
 N := 2^n = \#f + \#\overline{f}
\end{equation*}
which can be alternatively expressed as 
\begin{equation*}
 N = \bra{+}^{\otimes n}\ket{\psi_f} + \bra{+}^{\otimes n}\ket{\psi^\perp_f}.
\end{equation*}
Then evaluating the limit in \eqref{eqn:lowtemp-count} is equivalent to evaluating
the quantity $\bra{+}^{\otimes n}\ket{\psi_f}$, which we show has
desirable properties when expressed as a tensor contraction.  

As $\psi_f$ can be viewed as the state of a physical system, we can consider how
properties of tensor networks arise from motivations in physics. Penrose related the value of a tensor network contraction to its physical realizability:

\begin{theorem}[Penrose, \cite{Penrose67}]\label{thm:penrose}
 The norm of a spin network vanishes if and only if the physical situation it represents is
forbidden by the rules of quantum mechanics. 
\end{theorem}

\begin{example}[Example of Penrose's theorem]
 Consider a Bell state $\Phi^+ = \ket{00}+\ket{11}$.  The amplitude of the first party measuring $\ket{0}$ followed by the second party measuring $\ket{1}$ is zero. This vanishing tensor network contraction is given by $\braket{01}{\Phi^+}$.
\end{example}

Given a state $\ket{\psi}$, formed by a network of connected tensors, if $\braket{\psi}{\psi}$ vanishes, the network necessarily represents a non-physical quantum state by Penrose's theorem. In particular, the following example is a famous instance of non-physicality.

\begin{example}[Grandfather Paradox, \cite{2012PhRvA..86c0301M}]\label{ex:gfp}
   Consider the tensor $\psi=|01\rangle+|10\rangle$
contracted in the following network:
\begin{center}
\begin{tikzpicture}
 \draw[thick] (1.5,-.25) rectangle (2.5,-1.75);
\draw (2, -1) node{\textbf{$\psi$}};
\draw[thick] (1,-.5) arc (90:270:.5);
\draw[thick] (1,-.5) -- (1.5,-.5);
\draw[thick] (1,-1.5) -- (1.5,-1.5);
\end{tikzpicture}
\end{center}
The value of this network contraction is 
$$(\langle00|+\langle11|)(|01\rangle+|10\rangle)$$
$$=\langle00|01\rangle+\langle00|10\rangle+\langle11|01\rangle+\langle11|10\rangle=0$$
 This is the Grandfather
paradox: if you go back in time to kill your Grandfather, you could not have
been born to do the deed. Since the value of this tensor network is zero,
Theorem \ref{thm:penrose} tells us that this is indeed a paradox and cannot be realized
physically.
\end{example}

\subsection{Rephrasing \#SAT as a tensor network}\label{subsec:rephrasing}
We describe a common way to write a Boolean satisfiability problem as a tensor network \cite{JBCJ13}. Suppose we are given a
SAT formula $f$ and we wish to express it as a tensor $\psi_f$. Let $x_1,\dots,x_n$ be the variables appearing in this formula. We
have an open wire for each variable $x_i$. If a variable $x_i$ appears in $k$
clauses, we make $k$ copies of $x_i$ via the COPY-tensor $|0\rangle\langle
0|^{\otimes k}+|1\rangle\langle1|^{\otimes k}$. We call $k$ the \emph{degree} of
the COPY-tensor. It is depicted as a solid black dot in a tensor network.

 A Boolean gate (or clause) $\varphi:\{0,1\}^m\to\{0,1\}$ is expressed as the tensor
$$\sum_{\x\in\{0,1\}^m}{|\x\rangle\langle\varphi(\x)|}.$$ Then the COPY-tensor associated to each variable is connected to each tensor representing a clause which that variable appears in. This gives a tensor network with $n$ dangling wires corresponding to the variables and $c$ dangling wires, one for each clause. Lastly, we use a degree $c$ COPY-tensor to combine the wires of the clauses into a single output wire. The entire network can then be contracted to the tensor $\tilde{\psi}_f=\sum_{\x}{|x\rangle\langle f(\x)|}$ Then we define $\psi_f:=\tilde{\psi}_f|1\rangle$ and we call such a tensor a \emph{Boolean state}. 

\begin{remark}[Counting SAT solutions]\label{theorem:3-SAT}
Let $f$ be a SAT instance. Then the standard two-norm length squared of the corresponding Boolean state $\psi_f$ gives the number of satisfying assignments of the problem instance. 
\end{remark}

 The quantum state takes the form 
 \be
 \psi_f = \sum_{\x} \ket{\x}\braket{f(\x)}{1} = \sum_{\x} f(\x) \ket{\x}.
 \ee 
 We calculate the inner product of this state with itself viz 
 \be 
 ||\psi_f||^2=\sum_{\x, \y} f(\x) f(\y) \langle \x, \y\rangle = \sum_\x f(\x),
 \ee 
 which gives the number of satisfying inputs.  This follows since $f(\x)f(\y)=\delta_{\x\y}$.   
 We note that for Boolean states, the square of the two-norm
in fact equals the one-norm.

  We note that solving the counting problem for general
formula is known to be {\sf \#P}-complete \cite{valiant1979complexity}. The condition 
 \be 
 ||\psi_f|| > 0
 \ee 
 implies that the SAT instance $f$ has a satisfying assignment.  Determining if this condition holds for general Boolean states is 
 an {\sf NP}-complete decision problem.

\begin{figure}
 \centering
 \includegraphics[scale=1.2]{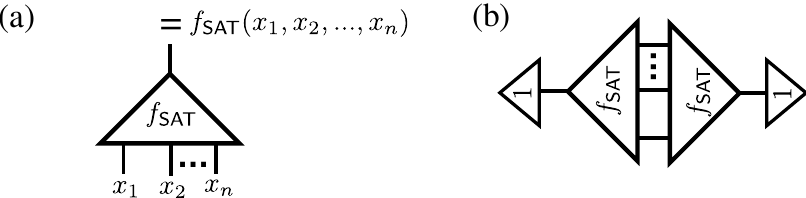}
 \caption{(a) A an abstract depiction of a Boolean state. (b) A Boolean state contracted with itself.}
 \label{fig:Booleanstates}
\end{figure}

 The tensor network contraction for the counting problem is depicted in Figure \ref{fig:Booleanstates}:  (a) gives
a network realization of the function and (b) is the contraction that represents
the norm of $f_{\textnormal{SAT}}$, where 1 has been selected as the desired
output of $f_{\textnormal{SAT}}$. If the value is greater than zero, the
instance has a solution.   

\section{Satisfiability of Certain SAT instances}\label{sec:sat-of-certain}

We give a few sufficient conditions for satisfiability. The first arises in Section \ref{sec:graphical}, namely if the Boolean state represents a
physical phenomenon.    

 \begin{corollary}
 SAT instances are solvable if and only if the corresponding contraction of a Boolean state with itself (as in Figure \ref{fig:Booleanstates} (b)) represents a physical situation allowed by the rules of quantum mechanics Theorem \ref{thm:penrose}.
\end{corollary}

\begin{example}
What does it mean that a SAT formula cannot be realized physically? Consider Example \ref{ex:gfp} above. We demonstrate how to view this as a SAT formula. We first put a variable on the wire. Then we change $\psi$ to $|01\rangle\langle1|+|10\rangle\langle1|$. Then we attach an output to $\psi$ and set it equal to 1. The picture is
\begin{center}
 \begin{tikzpicture}
   \draw[thick] (1.5,-.5) rectangle (2,-1.5);
\draw (1.75, -1) node{$\psi$};
\draw[thick] (2,-1) -- (2.25,-1);
\node[draw, fill=none, regular polygon, regular polygon sides=3,shape border rotate=30,inner sep=.15cm,thick] at (2.45cm,-1){};
\draw (2.43,-1) node{\textbf{\textnormal{1}}};
\node[draw, fill=none, regular polygon, regular polygon sides=3,shape border rotate=90,inner sep=.15cm,thick] at (.4cm,-1){};
\draw[thick] (.6,-1) -- (.8,-1);
\draw[fill=black] (.85,-1) circle (.1cm);
\draw (.4,-1) node{$x$};
\draw[thick] (1,-.75) .. controls (.77,-.75) and (.77,-1.25) .. (1,-1.25);
\draw[thick] (1,-.75) -- (1.5,-.75);
\draw[thick] (1,-1.25) -- (1.5,-1.25);
 \end{tikzpicture}
\end{center}

Since $x$ can be zero or one, it is represented by the tensor $\langle0|+\langle1|$. The dot is the COPY-tensor. Then we contract 
$$(\langle 0|+\langle1|)(|0\rangle\langle00|+|1\rangle\langle11|)(|01\rangle\langle1|+|10\rangle\langle1|)(|1\rangle)$$
$$=(\langle00|+\langle11|)(|01\rangle+|10\rangle)=0.$$
We see that this collapses to the grandfather paradox in Example \ref{ex:gfp}. But now we can interpret this as a SAT formula. Here $\psi$ corresponds to the formula $x\wedge\neg x$. 
\end{example}

Another class of SAT instances are always satisfiable simply by the structure of their tensor networks, namely trees where every variable is a leaf. These are more familiarly known as read-once functions.

 \begin{definition}[read-once]
 A function $f$ is called \emph{read-once} if it can be represented
 as a Boolean expression using the operations conjunction, disjunction and negation in
 which every variable appears exactly once. We call such a factored expression a read-once
 expression for $f$.
 \end{definition}

We call such a function ROF for short. The definition implies that there are no
COPY-tensors since every variable can only appear in one of the Boolean tensors
in the network. As such, it does indeed have a tree structure when expressed as
a tensor network.  

These formulas represent a special subclass of $r,s$-SAT, which is defined to be the decision problem for SAT formulas written in $r$-CNF form where each variable appears in at most $s$ clauses. Read once formulas, given in CNF form, represent the cases where $r$ is general and $s=1$. These problems have been studied classically and we mention two celebrated results:

\begin{theorem}[Tovey, \cite{tovey1984simplified}]\label{thm:tovey}
 Every instance of $r,r$-SAT is satisfiable.
\end{theorem}
\begin{theorem}[DuBois, \cite{dubois1990r}]\label{thm:dubois}
 If every instance of $r_0,s_0$-SAT is satisfiable, then $r,s$-SAT is satisfiable for $r=r_0+\lambda$ and $s\le s_0+\lambda\lceil s_0/r_0\rceil$, $\lambda\in\mathbb{N}$.
\end{theorem}

This implies that $r,s$-SAT is satisfiable for $s\le r$ by Theorem \ref{thm:tovey} and letting $\lambda=0$ in Theorem \ref{thm:dubois}. In particular, all ROF are satisfiable. Using tensor networks, we can give a very short proof of this.

\begin{definition}
 Let $f:\{0,1\}^n\to\{0,1\}$ be a Boolean function. It is represented by the tensor $\psi_f=\sum_{\x\in\{0,1\}^n}{|\x\rangle\langle f(\x)|}$. We call a tensor of the form $\psi^\dagger_f\psi_{f}$ a \emph{diagonal map}.
\end{definition}

\begin{lemma}\label{lem:diagmaps}
 For a Boolean function $\psi$, $\psi_f^\dagger \psi_f
=\#f^{-1}(0)|0\rangle\langle0|+\#f^{-1}(1)|1\rangle\langle 1|$ where
$\#f^{-1}(b)$ denotes the size of the pre-image of $b$.  Furthermore, if
$\psi_f$ is not the zero function, we can normalize it to get a tensor
$\zeta_f$ such that $\zeta^\dagger_f \zeta_f=\textnormal{id}$.
\end{lemma}
\begin{proof}
 $$\psi_f^\dagger\psi_f=\sum_{\x,\y}{|f(\x)\rangle\langle \x|\y\rangle\langle f(\y)|}$$
$$\;=\sum_{\x}{|f(\x)\rangle\langle f(\x)|}$$
$$\qquad\qquad\qquad\;\;\,=\#f^{-1}(0)|0\rangle\langle0|+\#f^{-1}(1)|1\rangle\langle 1|.$$

Now let $\zeta_f=\sum_\x{\sqrt{\#f^{-1}(f(\x))}^{-1}|\x\rangle\langle f(\x)|}$. It is clear that $\zeta_f$ has the desired property.
\end{proof}

For a Boolean state representing an ROF, we replace every gate with its
normalization. This gives positive weights to the different assignments
of variables. The contraction then sums up the weights of the satisfying
assignments. It is clear that normalizing a ROF does not change the fact that
the norm will be zero if and only if it has no satisfying solution. The
resulting scalar, however, will no longer reflect the number of satisfying
solutions.

\begin{theorem}\label{thm:rof}
 Every ROF is satisfiable.
\end{theorem}
\begin{proof}
 
We first normalize every gate in the tree. As discussed above, this new tensor network will have norm zero if and only if it is unsatisfiable. The tensor network is given as:

\begin{center}
 \includegraphics[scale=.2]{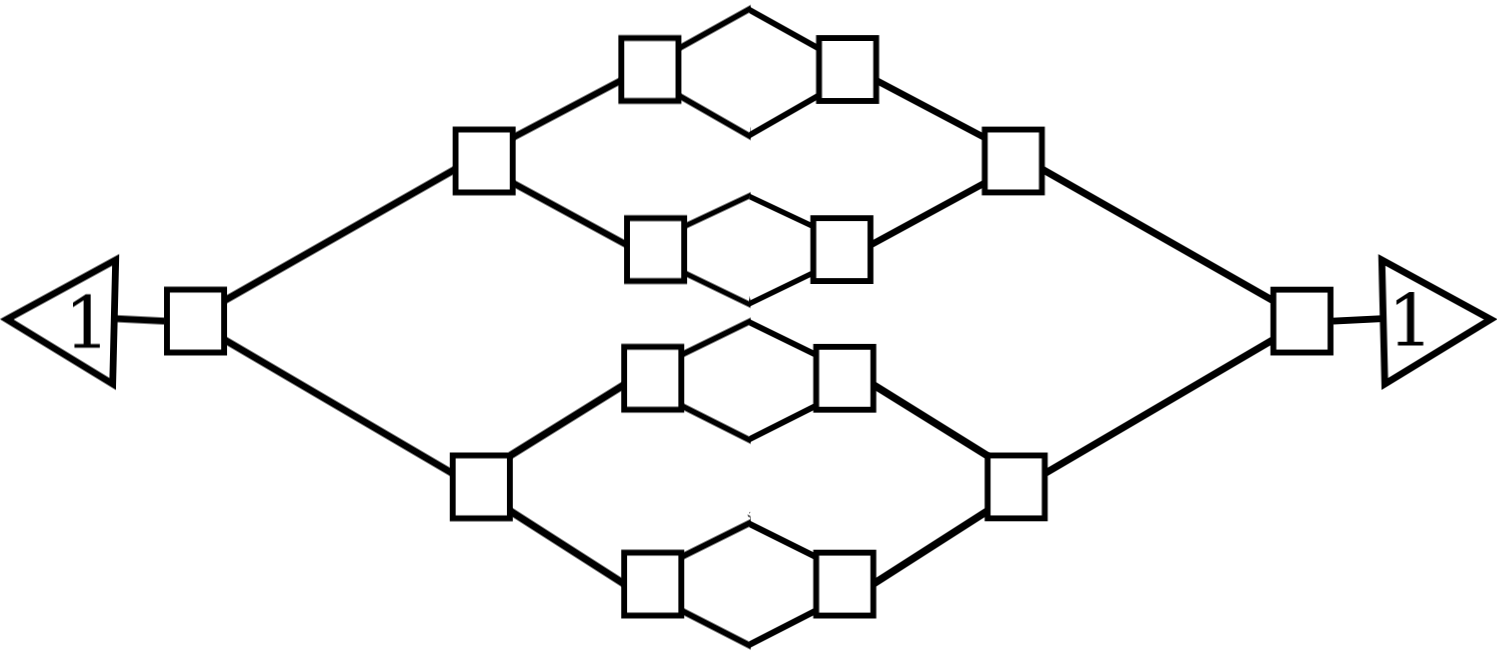}
\end{center}
We observe a series of nested diagonal maps. As any ROF does not allow the constant zero Boolean gate, then by Lemma \ref{lem:diagmaps} these maps
are equivalent to the identity. So we can successively collapse them until we
get the contraction $\langle 1|(\textnormal{id})|1\rangle=1$. So it is indeed satisfiable.

\end{proof}

\section{Efficient tensor contraction algorithm for high correspondence
networks}\label{sec:algorithm} 

Consider a tensor network with a known efficient, i.e. polynomial, contraction method. By inserting new wires and tensors into the network, one can construct a tensor network which is not necessarily easy to contract. Of interest are ways to reduce tensor networks to those that are computationally tractable in an efficient way.

In particular, we are interested in transforming Boolean states into others with polynomial time evaluations.  One approach is to restrict the number of tensors added to be at most $O(\log n)$, where $n$ is the number of variables.  

\begin{definition}[Tensor network correspondence] 
Let $v$ be the minimum number of vertices needed to be removed from the network $T$, with their incident edges becoming dangling wires, to transform it into a network $U$ known to be efficiently contractible. The \emph{correspondence} of $T$ with respect to $U$ is defined to be $1/v$.
\end{definition} 

\begin{example}[Low correspondence] 
It is well known that stabilizer circuits are efficiently contractible.  A polynomial number of stabilizer gates together with the addition of polynomially many single qubits gates, such as a phase gate, results in the stabilizer class of circuits being elevated to a universal set of quantum gates.  In general, these networks are as difficult to contract as any problem in QMA, and the correspondence of this family of networks is inversely proportional to some polynomial in the total number of single-qubit gates.  
\end{example}

For our purposes, we will be interested in restricting our definition of correspondence to the number of COPY-tensors needed to be removed to recover a tensor network that can be quickly evaluated. In this way, we can construct a wide class of counting problems that can be solved efficiently using tensor contraction algorithms.  Here we focus on the Boolean states that are close to tree tensor networks.   

It is well known that tree tensor networks are polynomially contractible in the number of gates \cite{halin1976s,lauritzen1988local,markov2008simulating}. For Boolean states, this includes the tensors $\langle0|+\langle1|$ placed on the variable wires.
\begin{definition}
 Given a rooted tree, a \emph{limb} is a sequence of vertices $v_1,\dots,v_k$ such that $v_2,\dots,v_{k-1}$ each only have one child and $v_{i+1}$ is a child of $v_i$.
\end{definition}
\begin{proposition}\label{prop:limbttn}
 An $n$-variable Boolean state that is a tree has $O(n)$ gates.
\end{proposition}
\begin{proof}
 A Boolean state has a natural root: the output bit. We argue that the limbs of a Boolean state are of length at most three. Take a limb $v_1,\dots,v_k$. Then consider $v_i$, $i\ne 1,k$. It is a unary operation on clauses, implying that it must either be the identity or negation operator. We can therefore assume that there is at most a single negation operator between $v_1$ and $v_k$ and nothing else. So at worst the tree is a perfect binary tree which has $O(n)$ gates.
\end{proof}
So we can conclude that the norm of a Boolean state that is a tree can be computed in polynomial time in the number of variables.

Let $X$ be a Boolean state. Corollary \ref{cor:c2sop} (in Appendix \ref{sec:tensor-defs}) allows us to remove a COPY-tensor, resulting in a sum of tensor networks:
\begin{center}
 \includegraphics[scale=.2]{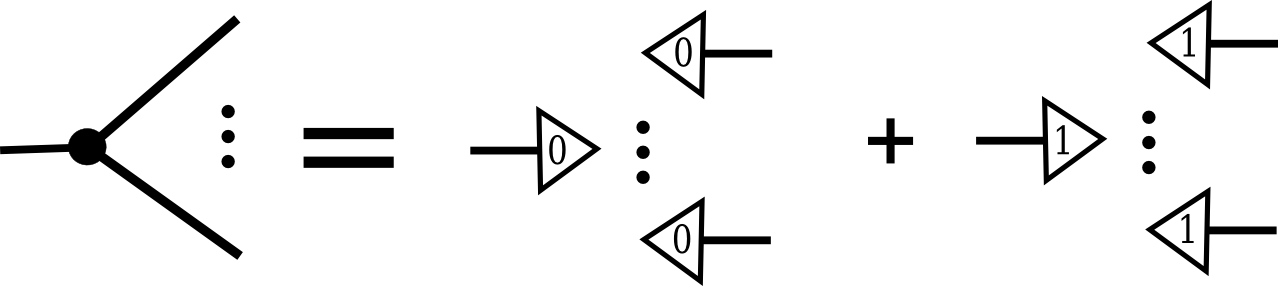}
\end{center}
We note that there are two summands, where either 0 or 1 is attached to the dangling wires resulting from the removal of the COPY-tensor. Given a tensor network $X$, if we remove a COPY-tensor, we denote the two summands $X_{0}$ and $X_1$. Furthermore, we have that $$\CC\{X,X\}=\sum_{i}{\CC\{X_i,X_i\}}.$$

We can now give an algorithm for contracting a tensor network, assuming we have an algorithm for contracting a tree tensor network in polynomial time. Let $C_1,\dots,C_m$ be the COPY-tensors appearing in $X$. Then $$\CC\{X,X\}=\sum_{i_1,\dots,i_m}{\CC\{X_{i_1,\dots,i_m},X_{i_1,\dots,i_m}\}}$$ where $X_{i_1,\dots,i_m}$ is the tree formed by removing all the COPY-tensors and assigning the value $i_k$ to the wires that were incident to $C_k$. So the algorithm for computing $\CC\{X,X\}$ computes the expression as a sum of contractions of trees.

\begin{theorem}[Upper bounding tensor contraction in terms of COPY-tensors]\label{thm:upperbound}
Given a tensor network as described in Subsection \ref{subsec:rephrasing}, the complexity of evaluating this network is $O((g+cd)^{O(1)} 2^c)$ where $c$ is the number of COPY-tensors, $g$ is the number of gates, and $d$ is the maximal degree of any COPY-tensor.
\end{theorem}
\begin{proof}
Contracting a tensor network is upper bounded by evaluating the expression $$\CC\{X,X\}=\sum_{i_1,\dots,i_m}{\CC\{X_{i_1,\dots,i_m},X_{i_1,\dots,i_m}\}}.$$ For each COPY-tensor removed, we double the number of summands. We also add a gate to each dangling wire created by removing a COPY-tensor. So the computation is bounded by summing over the contraction of $2^c$ trees, each of which can be contracted in time $O((g+cd)^{O(1)})$ \cite{markov2008simulating}.
\end{proof}

\begin{remark}
 We must mention that Theorem \ref{thm:upperbound} only applies when the Boolean state is given by a tensor network as described in Subsection \ref{subsec:rephrasing}. Otherwise, there are many tensor networks to describe a given Boolean state, many of which invalidate the statement of Theorem \ref{thm:upperbound}.
\end{remark}

\begin{corollary}[A class of efficiently solvable $\#$SAT instances]\label{cor:efficientSAT}
Suppose that a Boolean state with $n$ variables has inverse logarithmic correspondence with a tree tensor network and the maximal degree of the COPY-tensors is $O(n^\alpha)$. Then this Boolean state can be contracted in polynomial time.
\end{corollary} 
\begin{proof}
  Inverse logarithmic correspondence implies that if the SAT instance has $n$ variables, then it has $O(\log(n))$ COPY-tensors. By Theorem \ref{thm:upperbound}, the complexity of contraction is $O((g+n^\alpha\log(n))^{O(1)}2^{\log(n)})$. By Proposition \ref{prop:limbttn}, $g$ is $O(n)$, so the algorithm presented will contract the tensor network in polynomial time.
\end{proof}

Theorems \ref{thm:tovey} and \ref{thm:dubois}  only pertain to instances of $r,s$-SAT that are always satisfiable, they do not address the complexity of the corresponding counting problem.  For instance, read-twice monotone 3-CNF formulas are always satisfiable by Theorems \ref{thm:tovey} and \ref{thm:dubois}, but counting the number of solutions is $\#\P$-hard \cite{xia2007computational}.
%

Our results, in contrast, give conditions for the tractability of the counting problem associated to $r,s$-SAT. One of these conditions is not implicit in the definition of $r,s$-SAT, namely that we bound the number of variables that can appear in more than one clause. To restate Corollary \ref{cor:efficientSAT}, we have shown that \#$r,\textnormal{poly}(n)$-SAT is polynomial time countable if the number of COPY-tensors is bounded by $\log(n)$.

\subsection{Relation to Other Algorithms}

There have been several results relating the complexity of contracting tree tensor networks to various measures such as treewidth, clique-width, and branchwidth. An overview can be found in \cite{Nishimura2007solving}. The general type of result is that the time to solve a \#SAT instance is polynomial in the number of variables and exponential in the corresponding notion of width \cite{bacchus2003algorithms,fischer2008counting}.

 For example, in \cite{bacchus2003algorithms}, an algorithm is given that runs in time $n^{O(1)}2^{O(\omega)}$, where $\omega$ is the branchwidth of a certain hypergraph representing the SAT problem. Unfortunately, this result is not directly comparable to ours as it does not pertain to the branchwidth of the Boolean state but of a different structure.
 
We do wish to briefly discuss how our algorithm compares with algorithms based on treewidth of the Boolean state and the relationship between treewidth and the number of COPY-tensors. Contracting a tree tensor network $T$ takes time $g^{O(1)}\exp(\tn{tw}(T))$, where $\tn{tw}(T)$ is the treewidth of $T$ \cite{markov2008simulating}. More explicitly, we have the following results:
 \begin{theorem}[\cite{fischer2008counting}]\label{thm:treewidth}
  Given an $n$-variable SAT formula with treewidth $k$, there is an algorithm counting the number of solutions in time $4^k(n+n^2\log_2(n))$.
 \end{theorem}
\begin{theorem}[\cite{samer2006fixed}]
  Given a SAT formula with treewidth $k$, largest clause of size $l$, and $N$ the number of nodes of the tree-decomposition, the \#SAT problem can be solved in time $O(4^klN)$.
\end{theorem}

\begin{proposition}
 Given a Boolean state, let $k$ be the treewidth and $c$ the number of COPY-tensors. Then $k\le c$.
\end{proposition}
\begin{proof}
 If is well known that if $H$ is a subgraph of $G$, then $\tn{tw}(H)\le\tn{tw}(G)$. If we place the tensor $\langle0|+\langle1|$ on each variable wire and then compose with the COPY-tensors to get a tensor network $T$, a variable will correspond to a leaf if and only if it appears in exactly one clause. The treewidth of $T$ with $n$ variables,  $n'$ of which are not leaves, and $m$ clauses, is at most the treewidth of $K_{n',m}$, which is known to be $\tn{min}(m,n')$. This is because adding leaves to a graph does not increase the treewidth. On the other hand, $c=n'$.
\end{proof}

Note that if the number of variables is large with respect to the number of clauses, then $c$ may be much larger than the treewidth, $k$. However, the algorithms are comparable as long as $c$ is bounded by $2^{k}$ by Theorems \ref{thm:upperbound} and \ref{thm:treewidth}. We note that if the clause to variable ratio is extremely small or large, the tree width is small and the tensor network can be contracted efficiently.  In the critical case when $c\approxeq k$, our algorithm runs \emph{exponentially faster} in $k$. There is the added advantage that, unlike treewidth, calculating $c$ is not $\NP$-complete. As such, $c$ is an attractive estimate for the complexity of the counting problem. However, the tradeoff is that it is a cruder measurement.

\section{R{\'e}nyi Entropy and Complexity}\label{sec:renyi}

 R{\'e}nyi entropies \cite{Renyi2} have recently been proposed as an
indicator of the complexity of counting problems \cite{CM13}.  The R{\'e}nyi
entropies are attractive as they can be viewed
as providing a basis to express entanglement monotones
\cite{biamonte2012invariant} and
relate to the free-energy of a physical system \cite{2011arXiv1102.2098B}. 
We look at the question pertaining to the use of R{\'e}nyi entropies as a measurement of counting complexity. As we are interested in properties of tensor networks as measures of complexity, a natural question is the efficiency of determining the R{\'e}nyi entropy of a given state. Many proposed measures of complexity are difficult to determine. As we noted in Section \ref{sec:algorithm}, the number of COPY-tensors can be easily computed, whereas treewidth cannot. Another issue arises from the fact that  R{\'e}nyi entropies are defined for physical states, which by Theorem \ref{thm:penrose} does not include non-satisfiable instances of \#SAT.

In Section \ref{sec:graphical}, we described how to assign a Hamiltonian to a Boolean state. We choose a partition for the rows and columns into two disjoint subsystems, $A$ and $B$. This is called a \emph{bipartition} of the Hamiltonian, denoted $A:B$. The R{\'e}nyi entropy of order $q$ with respect to a bipartition $A:B$ is defined to be
\begin{equation*}
 H_q^{AB} (\rho) = \frac{1}{1-q} \ln \Tr \{\rho^q\} = \frac{1}{1-q} \ln \sum_i
\lambda_i^q, 
\end{equation*}
\begin{equation*}
 \rho = \frac{1}{Z_0} \Tr_{A}\{ \ket{\psi}\bra{\psi} \}.
\end{equation*}
Here, $Z_0=\Tr\{|\psi\rangle\langle\psi|\}$. The case $q\to0$ gives the rank of
the bipartition and the positive sided limit $q
\rightarrow 1$ recovers the familiar von Neumann Entropy $H^{AB}_{q\rightarrow 1}
(\rho) = - \Tr\{ \rho \ln \rho\}$.  

In terms of the counting problem, note that
the state $|\psi\rangle$ could be given by the zero vector, in which case it is not
possible to define the state $\rho$: so more generally we can consider the
unnormalized variant of $\rho$, which is defined as  $\rho =
\Tr_{A}\{\ket{\psi}\bra{\psi}\}$.
This has the following computational significance. 

\begin{theorem}[ R{\'e}nyi entropy implies satisfiability]\label{thm:renyi1}
For any bipartition, deciding if the R{\'e}nyi entropy is defined is NP-hard.
\end{theorem}
\begin{proof}
Given an unnormalized Boolean state $\rho$,
\begin{equation*}
 \rho = \Tr_{A}\{ \ket{\psi}\bra{\psi} \}
\end{equation*}
for some bipartition $A:B$, $H^{AB}_q (\rho)$ takes a finite value if and only if the Boolean
state is satisfiable. Note that this is independent of choice of bipartition as $\ln \Tr \{\rho^q\}$ is undefined if $|\psi\rangle=0$. Thus determining if the R{\'e}nyi entropy is defined is $\NP$-hard. 
\end{proof}

We can then consider instead the promise problem called Unambiguous-SAT: Given a SAT instance promised to have at most one solution, determine if it is satisfiable. Unambiguous-SAT is still a hard problem.

\begin{theorem}[Valiant-Vazirani, \cite{valiant1986np}]
If there is a polynomial time algorithm for solving Unambiguous-SAT, then $\NP=\sf{RP}$.
\end{theorem}
\begin{theorem}[R{\'e}nyi Entropy Reduction to Unambiguous-SAT]\label{thm:renyi2}
Suppose you have a Boolean state with the promise that the Renyi entropies are all undefined or zero. Then deciding if the R{\'e}nyi Entropies are undefined or zero is as hard as Unambiguous-SAT.
\end{theorem}
\begin{proof}
 The R{\'e}nyi entropies are zero if and only if the density operator is a product state, which implies that the corresponding Boolean state has a unique solution. It is undefined if and only if it has no solution as previously discussed. Therefore, if we are promised the R{\'e}nyi entropies are zero or undefined for all partitions, we are promised that there is at most one solution. This is precisely Unambiguous-SAT.
\end{proof}

What we see is that while R{\'e}nyi entropies may encode information about the complexity of a counting problem, they are unfortunately not easy to compute. Furthermore, they may not even be defined and determining this is computationally challenging.

\section{Discussion}\label{sec:discussion}
The tensor network approach offers several
straightforward proofs related to counting problems---including a simple proof
of a the $r$,1 instance of the Dubois-Tovey
theorem. In addition to this, a core result of this study was presenting a
tensor contraction algorithm for $n$-variable \#SAT instances with complexity
$O(n^{O(1)}2^c)$, where $c$ is the number of COPY-tensors in the
network, each with polynomial degree.  This provides a class of counting problem instances which can be
solved
efficiently when their tensor network expression has at most $O(\log c)$ COPY-tensors.

One can view this
work as a continuation of the ideas of classically simulating a quantum
algorithm \cite{JBCJ13} to solve a given computational task, which might provide
methods which are both quick to develop and scale favorably. Our results were all inspired and conducted in the language
of tensor networks.

\section*{acknowledgments}
We thank Tobias Fritz and Eduardo Mucciolo for providing feedback.  JDB acknowledges financial support from the Fondazione Compagnia di San Paolo
through the Q-ARACNE project and the Foundational Questions Institute (FQXi,
under grant FQXi-RFP3-1322).  JM and JT acknowledges the NSF (under grant NSF-1007808)
for financial support.

 \bibliographystyle{unsrt}
\bibliography{qc}

\newpage 
\appendix

\section{Properties of Tensor Contraction}\label{sec:tensor-defs} 

Although the graphical manipulation of tensors (e.g.~algebraic rewrite rules)
has been axiomatized by several communities, 
basic theorems and techniques of contraction are less systematically described. 
Networks have seen little effort toward their systematic formulation.  
Hence we present the minimal properties of tensor contraction we relied on to
provide the main results in the paper.  

We assume that the reader has some basic familiarly with tensor networks and thus
focus on the properties of contraction properties which we used.  This section should hopefully be self
contained.  For additional background information, see
e.g.~\cite{biamonte2012invariant}.	

The {\em value} of a tensor network is obtained by performing all possible
partial contractions diagrammed by the network (the order
of contraction does not affect the value). 
If we allow one or more of the tensors in the network to vary, and the network
has no dangling wires, we obtain a multilinear function from the direct sum over
vertices of the spaces $\bigotimes_{e\in\mathcal{N}(v)}{V(e)}$ to the complex
numbers. When there are dangling wires, the multilinear function's codomain
is the tensor product of the vector spaces corresponding to the dangling edges.

We denote such a function by $\CC$ when the tensor network defining it is
understood.  Then multilinearity can be rephrased as follows.
\begin{remark}[Linearity of tensor contraction] 

Tensor contraction is linear in its component tensors provided each appears only
once: $\CC\{A_{v_1}, \dots, \alpha A_{v_k}+A_{v_k}', \dots, A_{v_n}\} = \alpha
\CC\{A_{v_1}, \dots, A_{v_k}, \dots, A_{v_n}\} + \CC\{A_{v_1}, \dots, A_{v_k}',
\dots, A_{v_n}\}.$
and then $A\mapsto kA$ we readily find that the contraction becomes
$\CC\{A'\}+\CC\{B\}$ and $k\cdot \CC\{A\}$ respectively.  
\end{remark}

Several of our results rely on transforming  contracted sub-networks, into a
sum-over products, and vice versa.  
In other words, removing a contracted tensor (or part of a contracted tensor)
and replacing it by a sum of contracted tensors, where each term contains
contractions that are multiplied together.  The stage for this can be set by
stating several properties of tensor contractions.  We will use these properties
to reshape networks, such that their transformed geometry is known to be
efficiently contractible.

\begin{theorem}[COPY-tensors as a resolution of identity]\label{theorem:tnres}
The following sequence of graphical rewrites hold.  
\begin{center}
 \includegraphics[width=0.65\textwidth]{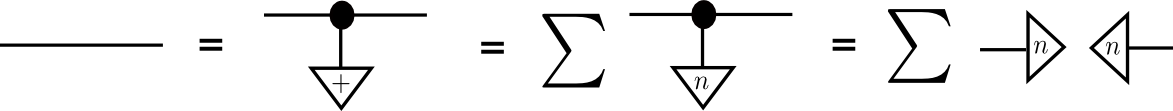} 
\end{center}
\end{theorem}

\begin{proof}
In the above figure, on the left, we abstractly depict a otherwise arbitrary
tensor network, by showing only one single wire.  The unit for the COPY-tensor
is the plus state $\ket{+}=\sum{|n\rangle}$. Contracting $\ket{+}$ against the
COPY-tensor therefore gives $\sum_{i,n}{|i\rangle\langle i|\langle
i|n\rangle}=\sum_{n}{|n\rangle\langle n|}$ which is precisely the identity.  The
last equality follows directly from the linearity of tensor contraction.
\end{proof}

This procedure can be iterated over many wires. In particular, if we have a gate
in a tensor network, we can remove it from the network by the above method
resulting in a nested sum of tensor networks. Of course the number of tensor
networks in the summand grows exponentially in the number of wires of the gate
removed. 

Some counting algorithms are only polynomial time on planar graphs, and this
cutting procedure can be used to deal with handles, yielding a running time
factor exponential in the genus of the graph \cite{bravyi2009contraction}.

\begin{corollary}[Contraction to sum of products transform]\label{cor:c2sop}
 Given a tensor $\Gamma^{...ijk}_{~~~~lmn...}$ in a fully contracted network
(e.g.\ a network without open legs), the following graphical identities, which are equal, transforms
the contraction, to a sum over products.  
\begin{center}
 \includegraphics[width=0.65\textwidth]{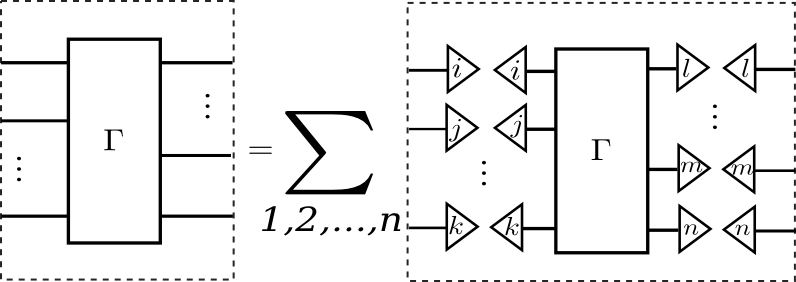} 
\end{center}
The comprising rectangle above is meant as an abstraction depicting a fully
contracted but otherwise general tensor network.  
\end{corollary}

As mentioned, the contraction to sum of products transform in Theorem
\ref{cor:c2sop} as well as the tensor network resolution of
identity in Theorem \ref{theorem:tnres} will be useful to transform networks
into those known to be efficiently contractible. 

\begin{theorem}[The Cauchy-Schwarz inequality]\label{thm:ineq}
Given a contracted network and a partition into two halves $x$, and $y$. 
Writing the contraction as 
$\CC\{x,y\}$ the following inequality holds.  
$$
\CC\{x,y\}^2\leq \CC\{x,x\}\cdot \CC\{y,y\}
$$ 
with graphical depiction.  
\begin{center}
 \includegraphics[width=0.6\textwidth]{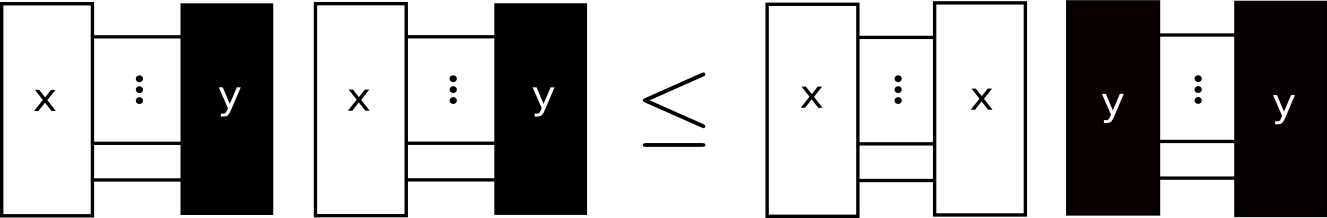} 
\end{center}
\end{theorem} 

\begin{proof}
 By the linearity of tensor contraction, we arrive at an abstract form of the 
 Cauchy-Schwarz inequality, with equality in the contraction if and only if $x =
\alpha \cdot y$.  This leads 
 directly to the concept of an angle between tensors, 
 $$ 
 \cos \theta_{xy} = \frac{\CC\{x,y\}^2}{\CC\{x,x\}\cdot \CC\{y,y\}}
 $$
 where the right side is either real valued, or we take the modulus.  
\end{proof}

\end{document}